\documentclass[english]{llncs}

\usepackage{amssymb}
\usepackage[english]{babel}
\usepackage{amsmath,mathtools,tikz}
\usepackage{graphicx}
\usepackage{color}
\usepackage[utf8]{inputenc}
\usepackage{subfigure}
\usepackage{comment}
\usepackage{float}

\usepackage{todonotes}
\usepackage{tikz}

\newcommand{\comm}[1]{}

\newcommand{\N}{\mathbb{N}}
\newcommand{\skin}{\mathsf{skin}}

\begin{document}

\title{On the number of useful objects in P systems with active membranes\thanks{Ministry of Human Capacities, Hungary grant 20391-3/2018/FEKUSTRAT is acknowledged. Szabolcs Iv\'an was supported by the J\'anos Bolyai Scholarship of the Hungarian Academy of Sciences.}}

\author{Zsolt Gazdag \and Károly Hajagos\and Szabolcs Iván}
\institute{Institute of Informatics\\
University of Szeged, Szeged, Hungary\\\email{\{gazdag,hajagos,szabivan\}@inf.u-szeged.hu}}

\maketitle

\begin{abstract}
In this paper we investigate the number of objects actually used in the terminating computations of a certain variant of polarizationless P systems with active membranes. The 
P systems we consider here have no in-communication rules and have no different rules triggered by the same object to manipulate the same membrane. 
We show that if we consider such a P system $\Pi$ and its terminating computation $\cal C$, then we can compute the result of $\cal C$ by setting a polynomial upper bound on the content of each region in $\cal C$.  
\end{abstract}

\section{Introduction}
The computational power of P systems with active membranes \cite{pau01} is widely investigated, mainly due to the fact that they can provide efficient solutions to computationally hard problems. The first such solutions include \cite{kriRam99,pau01,perJimSan03,zanFerMau01}, where $\mathbf{NP}$-complete problems were solved using only elementary membrane division. Non-elementary membrane division was also investigated and it turned out that using this type of rules, these P systems can already solve $\mathbf{PSPACE}$-complete problems \cite{alhMarPan03,sos03}.
Since then, many variants of P systems with active membranes were used to solve computationally hard problems, see, for example, \cite{alhPanPau04,alhPer07,gaz04,gazKol12,PanAlh06,PanAlhIsd04,Perez-Jimenez2005,PerRomSan06,sosRod07}. A recent survey can be found in \cite{sos19}. 

In \cite{alhPer07} the authors considered such P systems where non-elementary membrane division was allowed but the use of polarizations was not, and it was shown that this variant is still powerful enough to solve the $\mathbf{PSPACE}$-complete QSAT problem. On the other hand, it is still open, whether $\mathbf{NP}$-complete problems can be solved without polarization using only elementary membrane division. Due to the famous conjecture of  Gh. P\u aun \cite{pau05}, the answer to this question is expected to be negative. To prove this conjecture however is challenging since even a polarizationless P system can produce exponentially many regions and exponentially many objects in the regions in linear time.

Nevertheless, there are several partial solutions to P\u aun's conjecture (which is often called the \emph{P conjecture} in the literature), see, for example,  \cite{gazKol19,gutPerRisRom06,lepManMauPorZan17,murWoo07,wooMurPerRis09}.
In these partial solutions, the authors investigate a certain class of P systems, called \emph{recognizer P systems} \cite{perJimSan03}. Recognizer P systems are common tools when P systems are used to decide problems. They have many useful properties, for example, all of their computations halt and they have two designated objects $yes$ and $no$ which are sent to the environment exactly in the last step of the system. These objects are the outputs of the system and are used to indicate the acceptance or rejection of the input. If we use such a system to decide a problem, we require the system to be confluent meaning that all of the computations with the same input should give the same output. More precisely, given an instance $I$ of a decision problem $L$, a recognizer P system $\Pi$ deciding $L$ should halt on $I$ with the correct output $yes$ or $no$, according to whether $I$ is a positive instance of $L$ or not. That is, even if $\Pi$ can have many different computations on $I$, all of these computations should halt with the same output. Therefore, if we want to tell the output of $\Pi$ on $I$, it is enough to simulate an \emph{arbitrary} computation of $\Pi$. This implies that we can assume that $\Pi$ has no different rules involving the same membrane and triggered by the same object. 
To see this, assume that $\Pi$ has the following two rules $r_1: [a\to u]_1$ and $r_2: [a]_1\to [b]_1[c]_1$ (that is, $r_1$ and $r_2$ are usual evolution and division rules, respectively, with the usual specification of $a,b,c$, and $u$). Clearly, whenever $\Pi$ can apply $r_2$, it can apply $r_1$, too. Therefore, if we drop $r_2$ from $\Pi$, then $\Pi$ still has at least one computation which gives the correct output, for every possible input. According to this, to give a polynomial upper bound on the running time of recognizer P systems, it is enough to consider only certain computations of them, or even we can assume that these systems have the syntactic restrictions described above.
Similar observations are used in some of the partial solutions of the P conjecture in order to simplify the considered P systems. 

Another concept that is frequently arises in this research line is that of the  \emph{dependency graph} \cite{corGutPerRis05} and its variants. Roughly, these graphs describe how an object (or in same cases a configuration) evolves when certain rules are applied on it. For example, in \cite{gazKol19}, object division trees, a restricted variant of dependency graphs, are used to follow the evolution of objects under the application of division rules. 
These graphs are usually used to find such computations of a P system that can be simulated efficiently. For example, in  \cite{lepManMauPorZan17} and \cite{wooMurPerRis09}, the simulated computations are such that only a reasonable small part of them has to be represented in order to determine the result of these computations. In \cite{gazKol19}, the object division trees are used to find such computations which can be simulated by polynomial multiplication.

In the above mentioned partial solutions of the P conjecture, P systems with restricted initial membrane structure were investigated. In this paper we consider P systems with arbitrary initial membrane structure. Our P systems have no in-communication rules and have no different rules involving the same membrane and triggered by the same object. We show that for such a P system, we can compute the result of the terminating computations by setting a polynomial upper bound (depending of the length of the computation) on the content of each region. 
With our result, even if a P system $\Pi$ can produce exponentially many objects in the regions, it is possible to simulate $\Pi$ by keeping only a polynomial number of objects in each region (there can be exponentially many regions of $\Pi$, though).

In the proof of our result, we will need to use precisely such well known notions of Membrane Computing as \emph{maximal parallelism} and the \emph{computation step} of a P system. Thus we will give their formal definitions in the first part of the paper. Moreover, it will be convenient for us to treat a P system such that the rules of the system and the configurations the rules are working on are separated. Thus we will define \emph{membrane grammars}, which consist of similar rules as P systems except that evolution rules have the form $[a]_\ell\to [u]_\ell$ ($\ell$, $a$, and $u$ are specified as usual). Moreover, we will define \emph{membrane configurations} which are nonempty, finite, rooted, directed, edge-labeled trees. The nodes of such a tree will represent the regions of a membrane configuration, while the labeled edges are the membranes between the regions.

The paper is structured as follows. First, we give the necessary notions and notations in the next section. Then, in Section \ref{results}, we give the main results of the paper. We discuss the possible extensions and applications of our results in Section \ref{conclusion}.

\section{Preliminaries}
 
In this section, we introduce the necessary notions and notations. In particular, we define the notion of membrane grammars which is a novel representation of polarizationless P systems with active membranes. Nevertheless, we assume
that the reader is familiar with the basic concepts of membrane computing techniques (for a comprehensive guide see e.g. \cite{HandbookMC10}).

$\N$ stands for the set of natural numbers including zero, and for arbitrary $i,j\in\N, i\leq j,$ $[i,j]$ denotes the set $\{i,\ldots,j\}$. Furthermore, if $i=1$, then $[i,j]$ is denoted by $[j]$.

Let $O$ be an alphabet of objects and $H$ be a set of membrane labels. We assume that $H$ always contains the special label $\skin$. A \emph{membrane structure} is a triple $(V,E,L)$ where $(V,E)$ is a nonempty, finite, rooted, directed tree, having exactly one node in depth one, with edges directed towards the root, and $L:E\to H$ assigns labels to each edge, such that only the unique edge leading to the root can be labeled by the symbol $\skin$, and it has to be labeled by $\skin$. Edges are called \emph{membranes} of the structure, the nodes are its \emph{regions}. The root is also called the \emph{environment}. For each non-environment region $x\in V$, the outgoing edge $(x,y)\in E$ towards the parent of $x$ is called the outer membrane of the region, the edges directed into $x$ are called the inner membranes of $x$. We can assume that initially, each membrane has its unique label. The regions that have only an outgoing edge are the \emph{leaves} of the structure, and a membrane between a leaf and its parent is called an \emph{elementary membrane}.

A \emph{membrane configuration} is a tuple $(V,E,L,\omega)$ where $(V,E,L)$ is a membrane structure and $\omega:V\to O^*$ is a function which assigns a finite word of objects to each region. We view these words as \emph{multisets} of $O$, and also employ the functional notation: if $w\in O^*$ and $o\in O$, then $w(o)$ denotes the number of occurrences of $o$ in $w$. The empty word is denoted by $\varepsilon$. The difference $u-v$ of two words $u$ and $v$ is defined if for each $o\in O$, $u(o)\geq v(o)$, in which case $u-v$ is a word $w$ with
$w(o)+v(o)=u(o)$ for each $o\in O$. Sum of $u$ and $v$ is defined as their concatenation $u+v=uv$. If $t\in\N$, then the multiplication $t\cdot u$, for a word $u$, is a word $v$ with $v(o)=t\cdot u(o)$ for each $i\in O$. 

A \emph{membrane grammar} $G$ over $(O,H)$ is a finite set of rules of the following form:
\begin{align}
[a]_\ell&\to [u]_\ell&&\hbox{for some }a\in O,u\in O^*\hbox{ and }\ell\in H\tag{evolution}\\
[a]_\ell&\to b&&\hbox{for some }a,b\in O\hbox{ and }\ell\in H-\{\skin\}\tag{dissolution}\\
[a]_\ell&\to [b]_{\ell} [c]_{\ell}&&\hbox{for some }a,b,c\in O\hbox{ and }\ell\in H-\{\skin\}\tag{division}\\
[a]_\ell&\to []_\ell b&&\hbox{for some }a,b\in O\hbox{ and }\ell\in H\tag{out}\\
[]_\ell a&\to [b]_\ell&&\hbox{for some }a,b\in O\hbox{ and }\ell\in H-\{\skin\}\tag{in}
\end{align}

\noindent  The semantics of these rules will be described later, when we formally define a computation step of a membrane grammar. We just note here that it will be defined in the same way as in the case of P systems with active membranes \cite{pau01}. 
However, we impose some restrictions on the rules of a membrane grammar as we have discussed it in the Introduction. We require that for each $a\in O$ and $\ell\in H$, there is at most one rule with left-hand side $[a]_\ell$, and there is at most one rule with left-hand side $[]_\ell a$. We say that the pair $(o,\ell)$ is an evolution- / dissolution- / division- / out-pair, if there is an evolution / dissolution / division / out rule with left-hand side $[o]_\ell$, or an in-pair, if there is an in-rule with $[]_\ell o$ on its left-hand side. Division for non-elementary membranes is not allowed.

\begin{example}\label{ex-1}

Let $G$ be a membrane grammar over $(O,H)$ and let $(V,E,L,\omega)$ be a membrane configuration, where
\begin{itemize} 
\item $O=\{o_1,o_2\}$,
\item $H=\{\ell_1,\ldots,\ell_6,\skin\}$,
\item $V=\{env,s,x_1,\ldots,x_6\}$,
\item $E=\{(s,env),(x_1,s),(x_2,s),(x_3,x_1),(x_4,x_1),(x_5,x_2),(x_6,x_2)\}$,
\item and the rules of $G$ are:
\begin{align*}
&[o_1]_{\ell_1}\to []_{\ell_1}o_2&&\hbox{(out)}\\
&[o_2]_{\ell_1}\to o_1&&\hbox{(dissolution)}\\
&[o_1]_{\ell_2}\to[]_{\ell_2}o_1&&\hbox{(out)}\\
&[o_2]_{\ell_2}\to o_1&&\hbox{(dissolution)}\\
&[]_{\ell_3}o_1\to[o_1]_{\ell_3}&&\hbox{(in)}\\
&[]_{\ell_4}o_1\to[o_2]_{\ell_4}&&\hbox{(in)}\\
&[o_1]_{\ell_5}\to[o_2]_{\ell_5}[o_1]_{\ell_5}&&\hbox{(division)}\\
&[o_2]_{\ell_6}\to [o_1o_1]_{\ell_6}&&\hbox{(evolution)}
\end{align*}
\end{itemize}

\noindent Furthermore, let $L:E\to H$ be defined as follows:
\begin{align*}
& L(s,env)=\skin && L(x_1,s)=\ell_1 && L(x_2,s)=\ell_2\\
& L(x_3,x_1)=\ell_3 && L(x_4,x_1)=\ell_4 && L(x_5,x_2)=\ell_5\\& L(x_6,x_2)=\ell_6.
\end{align*}

\noindent Finally, we define $\omega:V\to O^*$ as follows:
\begin{align*}
& \omega(env)=\omega(s)=\varepsilon && \omega(x_1)=o_1o_1o_2 && \omega(x_2)=o_1o_2o_2\\
& \omega(x_3)=o_1o_2 && \omega(x_4)=o_1o_2 && \omega(x_5)=o_1o_1 \\
& \omega(x_6)=o_1o_2.
\end{align*}

\noindent Figure \ref{fig-conf} shows the membrane configuration $(V,E,L,\omega)$. The regions $x_3,x_4,x_5$ and $x_6$ are the leaves, therefore $\ell_3,\ell_4,\ell_5$, and $\ell_6$ are elementary membranes.

\begin{figure}
\centering \begin{tikzpicture}


\node[draw=black,circle,minimum size=3mm,label=env] (env) at (0,0.75) {};
\node[draw=black,circle,minimum size=7mm,label={[xshift=0mm,yshift=-11mm]$s$}] (skin) at (0,-1) {};
\node[draw=black,circle,minimum size=13mm,label={$x_1$}] (x1) at (-3,-2) {$o_1o_1o_2$};
\node[draw=black,circle,minimum size=13mm,label={$x_2$}] (x2) at (3,-2) {$o_1o_2o_2$};
\node[draw=black,circle,minimum size=13mm,label={[xshift=-6mm,yshift=-1mm]$x_3$}] (x3) at (-4.5,-4.5) {$o_1o_2$};
\node[draw=black,circle,minimum size=13mm,label={[xshift=6mm,yshift=-1mm]$x_4$}] (x4) at (-1.5,-4.5) {$o_1o_2$};
\node[draw=black,circle,minimum size=13mm,label={[xshift=-6mm,yshift=-1mm]$x_5$}] (x5) at (1.5,-4.5) {$o_1o_1$};
\node[draw=black,circle,minimum size=13mm,label={[xshift=6mm,yshift=-1mm]$x_6$}] (x6) at (4.5,-4.5) {$o_1o_2$};


\path
(skin) [->] edge node [left] {$\skin$} (env)
(x1) [->] edge node [above] {$\ell_1$} (skin)
(x2) [->] edge node [above] {$\ell_2$} (skin)
(x3) [->] edge node [left] {$\ell_3$} (x1)
(x4) [->] edge node [right] {$\ell_4$} (x1)
(x5) [->] edge node [left] {$\ell_5$} (x2)
(x6) [->] edge node [right] {$\ell_6$} (x2)
;

\end{tikzpicture}
\caption{A membrane configuration}
\label{fig-conf}
\end{figure}
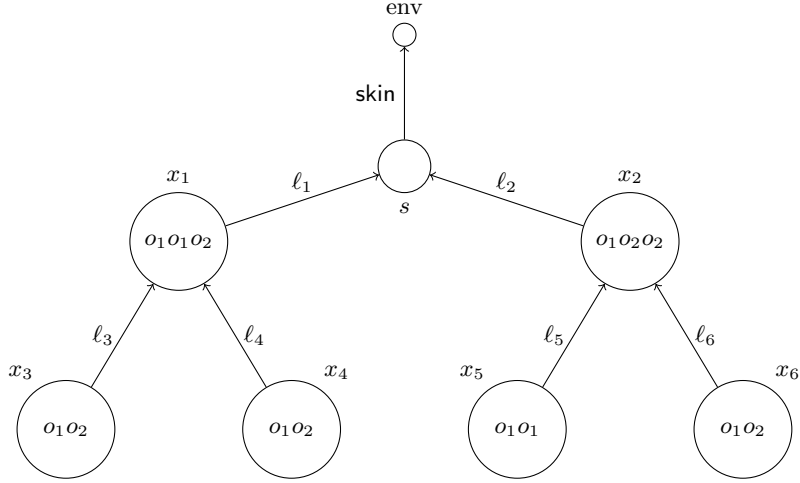
\end{example}

Next, we define a function on the edges (that is, on the membranes) of a configuration and call it \emph{viable transition}. It is used to assign objects to the membranes of a configuration according to the well-known notion of maximal parallelism.

\begin{definition}\label{defi-viable}
Let $G$ be a membrane grammar over $(O,H)$ and $C=(V,E,L,\omega)$ be a membrane configuration. The partial function $f:E\to O\times\{\uparrow,\downarrow\}$ is called a \emph{viable transition}, if it satisfies the following conditions:
\begin{enumerate}
\item[0.] If $f(x,y)=(o,\uparrow)$ for some $x,y\in V$ and $o\in O$, then there has to be a non-evolution rule with the left-hand side $[o]_{L(x,y)}$. Moreover, if $x$ is not a leaf, then this rule cannot be a division rule. If $f(x,y)=(o,\downarrow)$, then there has to be an in-rule with the left-hand side $[]_{L(x,y)} o$.
\item  For each region, there are at least as many objects inside the region who leave it. Formally, for each $x\in V$ and $o\in O$ we have $|\{(y,x)\in E:f(y,x)=(o,\downarrow)\}|$$+$$|\{(x,y)\in E:f(x,y)=(o,\uparrow)\}|$ is at most $\omega(x)(o)$. Note that the second set contains at most one edge, the one leading from $x$ to its parent.
\item For each region, if some object could use the outer membrane (and thus cannot evolve inside the membrane since, by assumption, $G$ has no evolution rules to do so), but the outer membrane is not used, then all those objects are occupied with in-rules. (Maximal parallelism.) Formally, for each edge $(x,y)$ with $f(x,y)$ being undefined, and for each $o\in O$ for which $(o,L(x,y))$ is a dissolution-, division- or out-pair, it has to be the case that $|\{(z,x)\in E:f(z,x)=(o,\downarrow)\}|=\omega(x)(o)$.
\item Similarly, for each region, if some object could use one of the inner membranes but it is not used, then all of those objects are occupied with evolution, or using the outer membrane or one of the inner membranes. Formally, for each edge $(y,x)$ with $f(y,x)$ being undefined, and for each $o\in O$ such that $(o,L(y,x))$ is an in-rule, one of the following cases has to hold, where $(x,z)$ is the outer membrane of the region $x$:
\begin{itemize}
\item $(o,L(x,z))$ is an evolution-pair;
\item or \[|\{(y',x)\in E:f(y',x)=(o,\downarrow)\}|~+~|\{(x,z):f(x,z)=(o,\uparrow)\}|~=~\omega(x)(o).\]
\noindent Again, observe that the second term of the sum is either zero or one.
\end{itemize}
\end{enumerate}
\end{definition}

Notice that viable transitions do not say explicitly which rules are attached to the membranes but, since $G$ has no different rules triggered by the same object for a membrane, we can tell which rules are about to use. A viable transition of the membrane grammar occurring in Example \ref{ex-1}
can be seen in Figure \ref{fig-viable}.

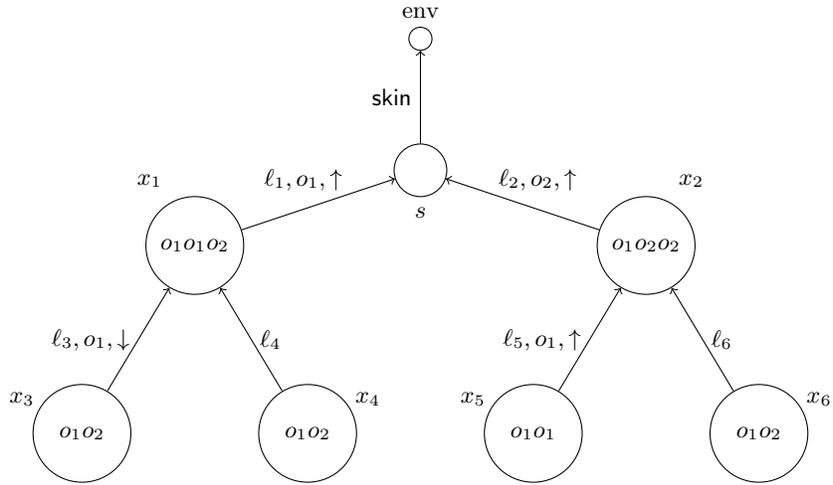
\begin{figure}
\centering \begin{tikzpicture}


\node[draw=black,circle,minimum size=3mm,label=env] (env) at (0,0.75) {};
\node[draw=black,circle,minimum size=7mm,label={[xshift=0mm,yshift=-11mm]$s$}] (skin) at (0,-1) {};
\node[draw=black,circle,minimum size=13mm,label={[xshift=-6mm]$x_1$}] (x1) at (-3,-2) {$o_1o_1o_2$};
\node[draw=black,circle,minimum size=13mm,label={[xshift=6mm]$x_2$}] (x2) at (3,-2) {$o_1o_2o_2$};
\node[draw=black,circle,minimum size=13mm,label={[xshift=-8mm,yshift=-4mm]$x_3$}] (x3) at (-4.5,-4.5) {$o_1o_2$};
\node[draw=black,circle,minimum size=13mm,label={[xshift=8mm,yshift=-4mm]$x_4$}] (x4) at (-1.5,-4.5) {$o_1o_2$};
\node[draw=black,circle,minimum size=13mm,label={[xshift=-8mm,yshift=-4mm]$x_5$}] (x5) at (1.5,-4.5) {$o_1o_1$};
\node[draw=black,circle,minimum size=13mm,label={[xshift=8mm,yshift=-4mm]$x_6$}] (x6) at (4.5,-4.5) {$o_1o_2$};


\path
(skin) [->] edge node [left] {$\skin$} (env)
(x1) [->] edge node [above=2pt,xshift=-2mm] {$\ell_1,o_1,\uparrow$} (skin)
(x2) [->] edge node [above=2pt,xshift=2mm] {$\ell_2,o_2,\uparrow$} (skin)
(x3) [->] edge node [left] {$\ell_3,o_1,\downarrow$} (x1)
(x4) [->] edge node [right] {$\ell_4$} (x1)
(x5) [->] edge node [left] {$\ell_5,o_1,\uparrow$} (x2)
(x6) [->] edge node [right] {$\ell_6$} (x2)
;

\end{tikzpicture}
\caption{A viable transition for the configuration of Example \ref{ex-1}. Here, the applicable rules are $[o_1]_{\ell_1}\to []_{\ell_1}o_2$, $[o_2]_{\ell_2}\to o_1$,	$[]_{\ell_3}o_1\to[o_1]_{\ell_3}$, $[o_1]_{\ell_5}\to[o_2]_{\ell_5}[o_1]_{\ell_5}$ and $[o_2]_{\ell_6}\to [o_1o_1]_{\ell_6}$.
}
\label{fig-viable}
\end{figure}

\begin{definition}
Let $C=(V,E,L,\omega)$ be a membrane configuration and $f:E\to O\times\{\uparrow,\downarrow\}$ be a viable transition for $C$. Then the \emph{computation step} $C\mathop{\vdash}\limits^f C'$ for the membrane configuration $C'$ is defined via several sub-steps discussed below. During the definition, we demonstrate the given steps by an example. We use the membrane grammar given in Example \ref{ex-1} and show the application of its rules on the configuration shown in Figure \ref{fig-conf} according to the viable transition given in Figure \ref{fig-viable}. The example is given via the corresponding figures below.
\begin{enumerate}
\item (Membrane attachments.) First, we remove those objects from the regions that use membranes. The new configuration $C_1=(V,E,L,\omega_1)$ is defined as follows:
	for each region $x\in V$ and object $o\in O$, let
	\begin{align*}
	\omega_1(x)(o)&=\omega(x)(o)\\
	&-|\{(y,x)\in E:f(y,x)=(o,\downarrow)\}|\\
	&-|\{(x,y)\in E:f(y,x)=(o,\uparrow)\}|.
	\end{align*}
Since $f$ is viable for $C$, these values are nonnegative for each $x$ and $o$.
	
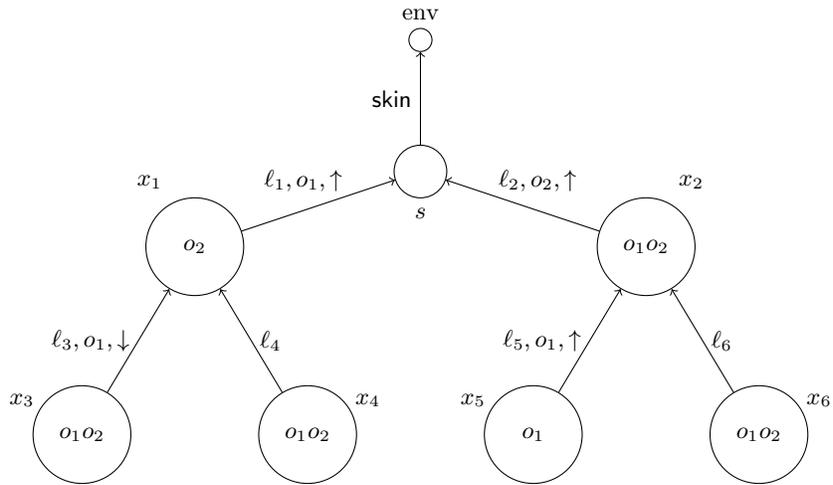
\begin{figure}
\centering \begin{tikzpicture}


\node[draw=black,circle,minimum size=3mm,label=env] (env) at (0,0.75) {};
\node[draw=black,circle,minimum size=7mm,label={[xshift=0mm,yshift=-11mm]$s$}] (skin) at (0,-1) {};
\node[draw=black,circle,minimum size=13mm,label={[xshift=-6mm]$x_1$}] (x1) at (-3,-2) {$o_2$};
\node[draw=black,circle,minimum size=13mm,label={[xshift=6mm]$x_2$}] (x2) at (3,-2) {$o_1o_2$};
\node[draw=black,circle,minimum size=13mm,label={[xshift=-8mm,yshift=-4mm]$x_3$}] (x3) at (-4.5,-4.5) {$o_1o_2$};
\node[draw=black,circle,minimum size=13mm,label={[xshift=8mm,yshift=-4mm]$x_4$}] (x4) at (-1.5,-4.5) {$o_1o_2$};
\node[draw=black,circle,minimum size=13mm,label={[xshift=-8mm,yshift=-4mm]$x_5$}] (x5) at (1.5,-4.5) {$o_1$};
\node[draw=black,circle,minimum size=13mm,label={[xshift=8mm,yshift=-4mm]$x_6$}] (x6) at (4.5,-4.5) {$o_1o_2$};


\path
(skin) [->] edge node [left] {$\skin$} (env)
(x1) [->] edge node [above=2pt,xshift=-2mm] {$\ell_1,o_1,\uparrow$} (skin)
(x2) [->] edge node [above=2pt,xshift=2mm] {$\ell_2,o_2,\uparrow$} (skin)
(x3) [->] edge node [left] {$\ell_3,o_1,\downarrow$} (x1)
(x4) [->] edge node [right] {$\ell_4$} (x1)
(x5) [->] edge node [left] {$\ell_5,o_1,\uparrow$} (x2)
(x6) [->] edge node [right] {$\ell_6$} (x2)
;

\end{tikzpicture}
\caption{After membrane attachment. Object $o_1$ of $x_1$ is attached to $\ell_1$ and one other $o_1$ is attached to $\ell_3$; $o_2$ of $x_2$ is attached to $\ell_2$ and $o_1$ of $x_5$ is attached to $\ell_5$.}
\label{fig-memb_att}
\end{figure}
	
\item (Evolutions.) Next, we apply the evolution rules. The new configuration $C_2=(V,E,L,\omega_2)$ is defined as follows: for each region $x\in V$, let $\omega_1(x)=o_1o_2\ldots o_m$ and $\ell=L(x,y)$ be the label of the outer membrane of $x$. Then define $\omega_2(x)$ as $u_1u_2\ldots u_m$ where \[u_i=\begin{cases}u&\hbox{if there is an evolution rule }[o_i]_\ell\to [u]_\ell\\ o_i&\hbox{otherwise}.\end{cases}\]

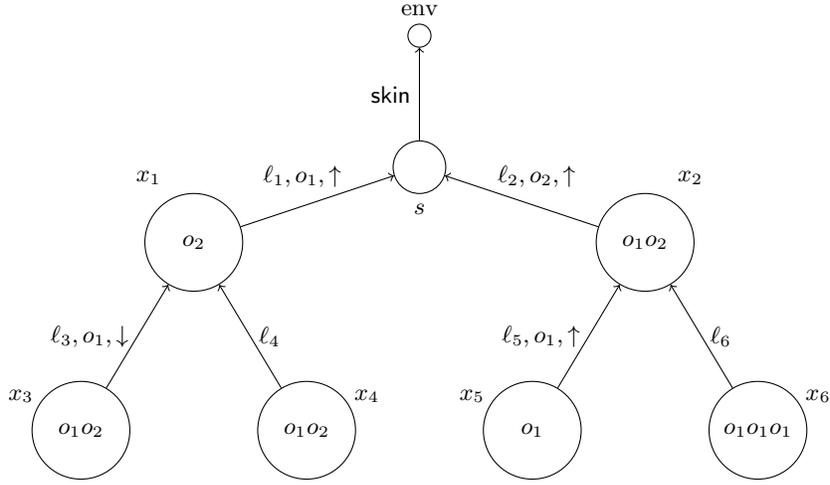
\begin{figure}
\centering \begin{tikzpicture}


\node[draw=black,circle,minimum size=3mm,label=env] (env) at (0,0.75) {};
\node[draw=black,circle,minimum size=7mm,label={[xshift=0mm,yshift=-11mm]$s$}] (skin) at (0,-1) {};
\node[draw=black,circle,minimum size=13mm,label={[xshift=-6mm]$x_1$}] (x1) at (-3,-2) {$o_2$};
\node[draw=black,circle,minimum size=13mm,label={[xshift=6mm]$x_2$}] (x2) at (3,-2) {$o_1o_2$};
\node[draw=black,circle,minimum size=13mm,label={[xshift=-8mm,yshift=-4mm]$x_3$}] (x3) at (-4.5,-4.5) {$o_1o_2$};
\node[draw=black,circle,minimum size=13mm,label={[xshift=8mm,yshift=-4mm]$x_4$}] (x4) at (-1.5,-4.5) {$o_1o_2$};
\node[draw=black,circle,minimum size=13mm,label={[xshift=-8mm,yshift=-4mm]$x_5$}] (x5) at (1.5,-4.5) {$o_1$};
\node[draw=black,circle,minimum size=13mm,label={[xshift=8mm,yshift=-4mm]$x_6$}] (x6) at (4.5,-4.5) {$o_1o_1o_1$};


\path
(skin) [->] edge node [left] {$\skin$} (env)
(x1) [->] edge node [above=2pt,xshift=-2mm] {$\ell_1,o_1,\uparrow$} (skin)
(x2) [->] edge node [above=2pt,xshift=2mm] {$\ell_2,o_2,\uparrow$} (skin)
(x3) [->] edge node [left] {$\ell_3,o_1,\downarrow$} (x1)
(x4) [->] edge node [right] {$\ell_4$} (x1)
(x5) [->] edge node [left] {$\ell_5,o_1,\uparrow$} (x2)
(x6) [->] edge node [right] {$\ell_6$} (x2)
;

\end{tikzpicture}
\caption{After evolutions. Inside $\ell_6$ the object $o_2$ got rewritten to $o_1o_1$.}
\label{fig-evo}
\end{figure}

\item (Movements.) Applying the in- and out-rules we get $C_3=(V,E,L,\omega_3)$ which is defined as follows: for each region $x\in V$ and $o\in O$, let
\begin{align*}
	\omega_3(x)(o)&=\omega_2(x)(o)\\
	&+|\{(y,x)\in E:f(y,x)=(o,\uparrow),(o,L(y,x))\hbox{ is an out-rule}\}|\\
	&+|\{(x,y)\in E:f(x,y)=(o,\downarrow),(o,L(x,y))\hbox{ is an in-rule}\}|.
\end{align*}

\begin{figure}
\centering \begin{tikzpicture}


\node[draw=black,circle,minimum size=3mm,label=env] (env) at (0,0.75) {};
\node[draw=black,circle,minimum size=7mm,label={[xshift=0mm,yshift=-11mm]$s$}] (skin) at (0,-1) {$o_2$};
\node[draw=black,circle,minimum size=13mm,label={[xshift=-6mm]$x_1$}] (x1) at (-3,-2) {$o_2$};
\node[draw=black,circle,minimum size=13mm,label={[xshift=6mm]$x_2$}] (x2) at (3,-2) {$o_1o_2$};
\node[draw=black,circle,minimum size=13mm,label={[xshift=-8mm,yshift=-4mm]$x_3$}] (x3) at (-4.5,-4.5) {$o_1o_1o_2$};
\node[draw=black,circle,minimum size=13mm,label={[xshift=8mm,yshift=-4mm]$x_4$}] (x4) at (-1.5,-4.5) {$o_1o_2$};
\node[draw=black,circle,minimum size=13mm,label={[xshift=-8mm,yshift=-4mm]$x_5$}] (x5) at (1.5,-4.5) {$o_1$};
\node[draw=black,circle,minimum size=13mm,label={[xshift=8mm,yshift=-4mm]$x_6$}] (x6) at (4.5,-4.5) {$o_1o_1o_1$};


\path
(skin) [->] edge node [left] {$\skin$} (env)
(x1) [->] edge node [above=2pt,xshift=-2mm] {$\ell_1$} (skin)
(x2) [->] edge node [above=2pt,xshift=2mm] {$\ell_2,o_2,\uparrow$} (skin)
(x3) [->] edge node [left] {$\ell_3$} (x1)
(x4) [->] edge node [right] {$\ell_4$} (x1)
(x5) [->] edge node [left] {$\ell_5,o_1,\uparrow$} (x2)
(x6) [->] edge node [right] {$\ell_6$} (x2)
;

\end{tikzpicture}
\caption{After movements. The membrane $\ell_1$ released $o_1$ as $o_2$ upwards, $\ell_3$ released $o_1$ as $o_1$ downwards.}
\label{fig-movements}
\end{figure}
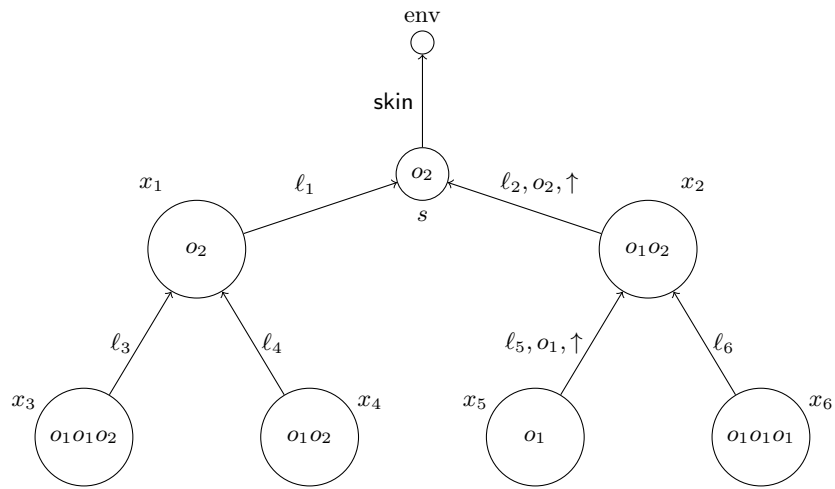

\item (Dissolutions.) Applying the dissolution rules, the current configuration is $C_4=(V_4,E_4,L_4,\omega_4)$ which we define as follows. Initially let $C_4=C_3$. Iterating from the leaves towards the root, we dissolve membranes step by step as follows: if for an edge $(x,y)\in E_4$,  $f(x,y)=(o_1,\uparrow)$ so that $(o_1,L_4(x,y))$ is a dissolution-pair with a rule $[o_1]_{L_{4}(x,y)}\to o_2$ for some $o_1,o_2\in O$, then we set $\omega_4(y)=\omega_4(y)+\omega_4(x)+o_2$. Moreover, for each $z\in V_4$ with $(z,x)\in E_4$, we add a new edge $(z,y)$ to $E_4$, set $L_4(z,y)=L_4(z,x)$,  $E_4=E_4-\{(z,x)\}$. Furthermore, we set $E_4=E_4-\{(x,y)\}$ and $V_4=V_4-\{x\}$. We repeat this process till we handled all the dissolution-marked membranes.

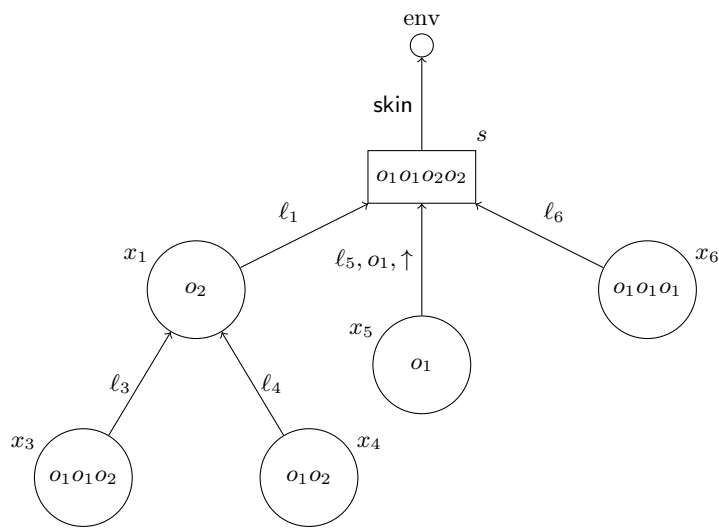
\begin{figure}
\centering \begin{tikzpicture}

\node[draw=black,circle,minimum size=3mm,label=env] (env) at (0,0.75) {};
\node[draw=black,rectangle,minimum size=7mm,label={[xshift=8mm,yshift=0mm]$s$}] (skin) at (0,-1) {$o_1o_1o_2o_2$};
\node[draw=black,circle,minimum size=13mm,label={[xshift=-8mm,yshift=-4mm]$x_1$}] (x1) at (-3,-2.5) {$o_2$};
\node[draw=black,circle,minimum size=13mm,label={[xshift=-8mm,yshift=-4mm]$x_3$}] (x3) at (-4.5,-5) {$o_1o_1o_2$};
\node[draw=black,circle,minimum size=13mm,label={[xshift=8mm,yshift=-4mm]$x_4$}] (x4) at (-1.5,-5) {$o_1o_2$};
\node[draw=black,circle,minimum size=13mm,label={[xshift=-8mm,yshift=-4mm]$x_5$}] (x5) at (0,-3.5) {$o_1$};
\node[draw=black,circle,minimum size=13mm,label={[xshift=8mm,yshift=-4mm]$x_6$}] (x6) at (3,-2.5) {$o_1o_1o_1$};


\path
(skin) [->] edge node [left] {$\skin$} (env)
(x1) [->] edge node [above=2pt,xshift=-2mm] {$\ell_1$} (skin)
(x3) [->] edge node [left] {$\ell_3$} (x1)
(x4) [->] edge node [right] {$\ell_4$} (x1)
(x5) [->] edge node [left] {$\ell_5,o_1,\uparrow$} (skin)
(x6) [->] edge node [above=2pt,xshift=2mm] {$\ell_6$} (skin)
;

\end{tikzpicture}
\caption{After dissolutions. The membrane $\ell_2$ got dissolved under the rule $[o_2]_{\ell_2}\to o_1$.}
\label{fig-dis}
\end{figure}

\item (Divisions.) Finally, $C'=(V',E',L',\omega')$ is defined as follows. Initially let $C'=C_4$. For each membrane $(x,y)\in E'$ such that $f(x,y)=(o,\uparrow)$ and $(o,L'(x,y))$ is a division-pair with a rule $[o]_{L'(x,y)}\to [o_1]_{L'(x,y)}[o_2]_{L'(x,y)}$ for some $o,o_1,o_2\in O$, let $V'=V'\cup\{x'\}$, where $x'$ is a new child of $y$ with $\omega'(x')=\omega'(x)+o_2$ and $L'(x',y)=L'(x,y)$. Then we set $\omega'(x)=\omega'(x)+o_1$.

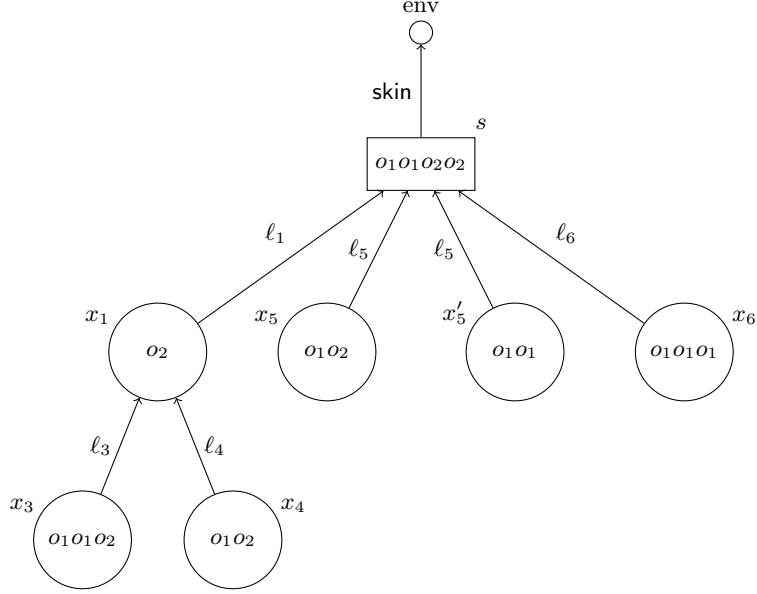
\begin{figure}
\centering \begin{tikzpicture}


\node[draw=black,circle,minimum size=3mm,label=env] (env) at (0,0.75) {};
\node[draw=black,rectangle,minimum size=7mm,label={[xshift=8mm,yshift=0mm]$s$}] (skin) at (0,-1) {$o_1o_1o_2o_2$};
\node[draw=black,circle,minimum size=13mm,label={[xshift=-8mm,yshift=-4mm]$x_1$}] (x1) at (-3.5,-3.5) {$o_2$};
\node[draw=black,circle,minimum size=13mm,label={[xshift=-8mm,yshift=-4mm]$x_3$}] (x3) at (-4.5,-6) {$o_1o_1o_2$};
\node[draw=black,circle,minimum size=13mm,label={[xshift=8mm,yshift=-4mm]$x_4$}] (x4) at (-2.5,-6) {$o_1o_2$};
\node[draw=black,circle,minimum size=13mm,label={[xshift=-8mm,yshift=-4mm]$x_5$}] (x5) at (-1.25,-3.5) {$o_1o_2$};
\node[draw=black,circle,minimum size=13mm,label={[xshift=-8mm,yshift=-4mm]$x'_5$}] (x5') at (1.25,-3.5) {$o_1o_1$};
\node[draw=black,circle,minimum size=13mm,label={[xshift=8mm,yshift=-4mm]$x_6$}] (x6) at (3.5,-3.5) {$o_1o_1o_1$};


\path
(skin) [->] edge node [left] {$\skin$} (env)
(x1) [->] edge node [above=2pt,xshift=-2mm] {$\ell_1$} (skin)
(x3) [->] edge node [left] {$\ell_3$} (x1)
(x4) [->] edge node [right] {$\ell_4$} (x1)
(x5) [->] edge node [left] {$\ell_5$} (skin)
(x5') [->] edge node [left] {$\ell_5$} (skin)
(x6) [->] edge node [above=2pt,xshift=2mm] {$\ell_6$} (skin)
;

\end{tikzpicture}
\caption{After divisions. The membrane $\ell_5$ became divided under the rule $[o_1]_{\ell_5}\to[o_2]_{\ell_5}[o_1]_{\ell_5}$.}
\label{fig-div}
\end{figure}				
	
The membrane configuration we end in up after this last step is $C'$ with $C\mathop{\vdash}\limits^{f}C'$. If there is such a viable $f$, then we write $C\vdash C'$.
\end{enumerate}
\end{definition}



\section{Results}\label{results}
Consider a membrane grammar $G$ and a terminating computation $C_1\vdash\ldots \vdash C_t$ of $G$. In this section, we show that if $G$ does not have in-rules, then $G$ actually uses at most $t$ copies of each object of each membrane of  $C_1$. In other words, we can apply a threshold $t$ on the content of the membranes in $C_1$ without affecting the result of the computation. First, we give a formal definition of applying a threshold on a multiset or a configuration.

When $w$ is a multiset over $O$ and $t\in \N$ is a threshold, then let $w|_t$ be the multiset with \[w|_t(o)=\begin{cases}
w(o)&\hbox{if }w(o)<t\\
t&\hbox{otherwise.}
\end{cases}\]
For a configuration $C=(V,E,L,\omega)$ and a number $t\in\N$, let $C|_t$ denote the configuration we get by applying the threshold $t$ on the content of each region of $C$, hence $C|_t=(V,E,L,\omega')$, where $\omega'(x)=\omega(x)|_t$ for each $x\in V$ (an example can be seen in Figure \ref{fig-conf-threshold}). 
We say that the multisets $w_1,w_2$ over $O$ are $t$-equivalent ($w_1\approx_t w_2$), if $w_1|_t=w_2|_t$.
Similarly, two membrane configurations $C_1$ and $C_2$ are $t$-equivalent, denoted $C_1\approx_t C_2$, if they have the same membrane structure $(V,E,L)$ and for each region $x\in V$, we have $\omega_1(x)\approx_t\omega_2(x)|_t$. Clearly, $C\approx_t C|_t$ since we got $C|_t$ by applying the threshold $t$ on the content of each region of $C$.

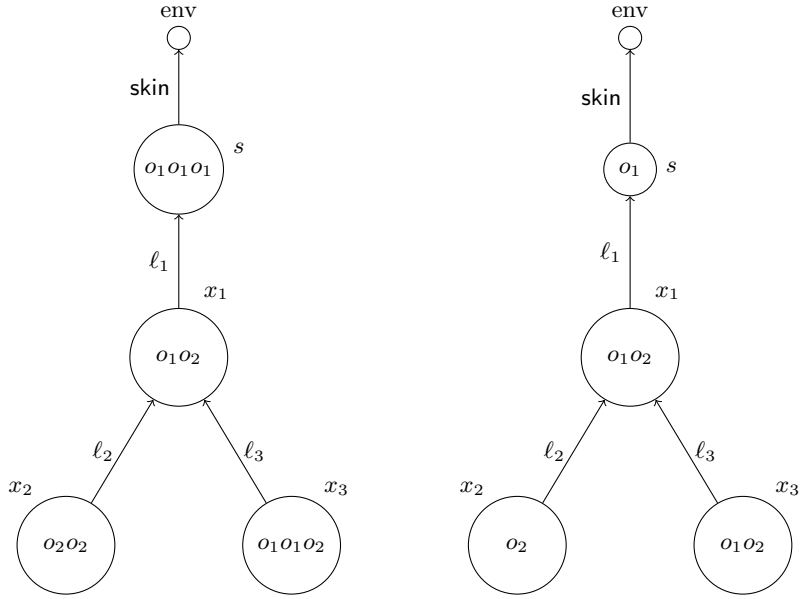
\begin{figure}
\centering \begin{tikzpicture}


\node[draw=black,circle,minimum size=3mm,label=env] (env) at (-3,0.75) {};
\node[draw=black,circle,minimum size=7mm,label={[xshift=8mm,yshift=-5mm]$s$}] (skin) at (-3,-1) {$o_1o_1o_1$};
\node[draw=black,circle,minimum size=13mm,label={[xshift=5mm,yshift=0mm]$x_1$}] (x1) at (-3,-3.5) {$o_1o_2$};
\node[draw=black,circle,minimum size=13mm,label={[xshift=-6mm,yshift=-1mm]$x_2$}] (x2) at (-4.5,-6) {$o_2o_2$};
\node[draw=black,circle,minimum size=13mm,label={[xshift=6mm,yshift=-1mm]$x_3$}] (x3) at (-1.5,-6) {$o_1o_1o_2$};

\node[draw=black,circle,minimum size=3mm,label=env] (env') at (3,0.75) {};
\node[draw=black,circle,minimum size=7mm,label={[xshift=5.5mm,yshift=-5mm]$s$}] (skin') at (3,-1) {$o_1$};
\node[draw=black,circle,minimum size=13mm,label={[xshift=5mm,yshift=0mm]$x_1$}] (x1') at (3,-3.5) {$o_1o_2$};
\node[draw=black,circle,minimum size=13mm,label={[xshift=-6mm,yshift=-1mm]$x_2$}] (x2') at (1.5,-6) {$o_2$};
\node[draw=black,circle,minimum size=13mm,label={[xshift=6mm,yshift=-1mm]$x_3$}] (x3') at (4.5,-6) {$o_1o_2$};


\path
(skin) [->] edge node [left] {$\skin$} (env)
(x1) [->] edge node [left] {$\ell_1$} (skin)
(x2) [->] edge node [left] {$\ell_2$} (x1)
(x3) [->] edge node [right] {$\ell_3$} (x1)

(skin') [->] edge node [left] {$\skin$} (env')
(x1') [->] edge node [left] {$\ell_1$} (skin')
(x2') [->] edge node [left] {$\ell_2$} (x1')
(x3') [->] edge node [right] {$\ell_3$} (x1')
;

\end{tikzpicture}
\caption{Membrane configurations $C$ and $C|_1$}
\label{fig-conf-threshold}
\end{figure}

\begin{theorem}\label{theo1}
Let $G$ be a membrane grammar over $(O,H)$. Assume there are no in-rules in $G$,
$C_1$ and $C_1'$ are $t$-equivalent configurations for some $t>0$, and $C_1\vdash C_2$. Then there exists a configuration $C_2'$ with $C_1'\vdash C_2'$ and $C_2\approx_{t-1}C_2'$.
\begin{proof}
Let us write in more detail $C_1\mathop{\vdash}\limits^{f}C_2$ and let $C_1=(V,E,L,\mu_1)$ and $C_1'=(V,E,L,\mu_1')$. We show that the same function $f$ is viable for $C_1'$ as well, and for the configuration $C_2'$ with $C_1'\mathop{\vdash}\limits^{f}C_2'$ we have $C_2\approx_{t-1} C_2'$.

As there are no in-rules, the conditions for viability become simpler. We check them one by one:
\begin{enumerate}
\item[0.] Since this condition depends only on the rules, it is satisfied for $C_1'$ as well.
\item Since $f$ is viable for $C_1$, we have that for each $x\in V$ and $o\in O$, $|\{(x,y)\in E:f(x,y)=(o,\uparrow)\}|$ is at most $\mu_1(x)(o)$. Since the cardinality of this set is either one or zero and $t>0$, we have that whenever $\mu_1(x)(o)$ is at least one, then so is $\mu'_1(x)(o)$, so this condition is satisfied.
\item As there are no in-rules, the set in the second condition of viability is empty, meaning $\mu_1(x)(o)=0$ for those edges and objects. Hence, by $t$-equivalence we get that $\mu_1'(x)(o)=0$ as well for these pairs $(x,o)$, thus this condition is also satisfied.
\item As there are no in-rules, this condition is also vacuously satisfied.
\end{enumerate}
Thus, $f$ is viable for $C_1'$ as well. Now let $C_2'$ be the configuration with $C_1'\mathop{\vdash}\limits^{f}C_2'$. We show that for each intermediate step in the definition of a computation step, the configurations are $(t-1)$-equivalent.
\begin{enumerate}
\item After the membrane attachment step, since there is no in-rule, $\omega_1(x)(o)$ either remains the same or decreases by one for each region $x$ and object $o$, depending on $f$. Thus, if $\mu_1(x)(o)=\mu_1'(x)(o)$, then they will be the same after attachment; if both of them are at least $t$,
then both of them will be at least $t-1$ after the attachment step.
\item Let $x$ be a region whose outer membrane is labeled by $\ell$. For each $o\in O$, let $u_o$ denote $u\in O^*$ if there is an evolution rule $[o]_\ell\to [u]_\ell$, and $o$ otherwise. After evolution, $\omega_2(x)$ will become $\mathop\sum\limits_{o\in O}\bigl(\omega_1(x)(o)\cdot u_o\bigr)$. Since if $t_1,t_2\geq t$, then $t_1\cdot u\approx_t t_2\cdot u$ for any multiset $u$, we get that since before evolution the (intermediate) configurations were $(t-1)$-equivalent, so are they after evolution.
\item Movements can only increase the contents of a region (in fact, the membrane attachment step is the only one decreasing the content), and by the same amount as they depend only on $f$. Clearly, $u\approx_{t-1}v$ implies $(u+o)\approx_{t-1}(v+o)$ for any object $o$ and multisets $u,v$ over $O$ (actually, if $u'\approx_{t-1}v'$ for any other pair of multisets, then $(u+u')\approx_{t-1}(v+v')$). Thus, after movements the configurations are still $(t-1)$-equivalent.
\item Dissolving a membrane $(x,y)$ with a rule $[o_1]_\ell\to o_2$ increases $\omega_4(y)$ by $\omega_4(x)+o_2$. Since, for arbitrary $(x,y)\in E$, $\mu_1(y)\approx_{t-1}\mu_1'(y)$ and $\mu_1(x)\approx_{t-1}\mu_1'(x)$ by the assumption, moreover $f$ dissolved $x$ according to the same dissolution rule in $C_1$ and $C_1'$ (thus the object on the right-hand side of this rule is the same), $(t-1)$-equivalence is retained after each dissolution.
\item After performing a division by applying a rule $[o]_\ell\to[o_1]_\ell[o_2]_\ell$, the membrane structure will be the same in both configurations. Moreover, if before the division the configurations are $(t-1)$-equivalent, then so they are after the division, and when we insert the two objects $o_1$ and $o_2$, the corresponding regions remain $(t-1)$-equivalent.
\end{enumerate}
Hence, for $C_2'$ we indeed have $C_2\approx_{t-1}C_2'$.
\end{proof}
\end{theorem}

Using Theorem \ref{theo1}, we can show that for a computation $\cal C$ which terminates in $t$ steps, we can give another computation $\cal C'$ such that the following holds. The number of each object in the first configuration of $\cal C'$ is bounded by $t$ and the environment at the end of both computations contains the same objects.
\begin{corollary}\label{cor1}
Let $\mathcal{C}=C_1\vdash\ldots\vdash C_t$ be a terminating computation for some $t>0$ and let $C'_1=C_1|_t$. Then there exists another terminating computation $\mathcal{C}'=C'_1\vdash\ldots\vdash C'_t$ with $C_t\approx_1 C'_t$.
\end{corollary}
\begin{proof}
We prove more: there exists a terminating computation $\mathcal{C}'=C'_1\vdash\ldots\vdash C'_t$ such that $C_i\approx_{t-i+1} C'_i$ for each $i\in[t]$.
We prove this by induction on $i$.

If $i=1$, then we have $C_1\approx_t C'_1$. If $i>1$, then by induction we have $C_{i-1}\approx_{t-(i-1)+1} C'_{i-1}$. Since $C_{i-1}\vdash C_i$, using Theorem \ref{theo1} we get that there exists a configuration $C'_i$ with $C'_{i-1}\vdash C'_{i}$ and $C_i\approx_{t-i+1}C'_i$.

Thus $C_t\approx_1 C'_t$, which concludes the proof of the statement. 
\end{proof}


\section{Conclusions}\label{conclusion}
Consider a membrane grammar $G$ over $(O,H)$ such that $G$ has no in-rules. By the iterated application of Corollary \ref{cor1}, we can simulate a terminating computation $C_1\vdash\ldots \vdash C_t$ of $G$ as follows. We compute a configuration sequence $D_1,D_2,\ldots, D_t$ such that, for each $i\in [t]$, $D_i\approx_{t-i+1} C_i$ and each membrane in $D_i$ contains only a polynomial number of objects in $|O|+t$. On the other hand, these configurations can contain exponentially many membranes. 
Nevertheless, we believe that our result can be used to give polynomial-time simulations of certain variants of polarizationless P systems with active membranes. For example, in \cite{gazKol19}, a novel method was given to simulate simple polarizationless P systems efficiently. The simulated P systems are such that they have only one membrane in the skin membrane at the beginning of the computation. To extend the simulation given in \cite{gazKol19} to P systems having an arbitrary membrane structure, we need that those membranes that become elementary during the computation contain polynomially many objects. By the results of this paper, we have a chance to achieve this property. The elaboration of the details is a topic for future work. 

Moreover, we think that the proof of Theorem \ref{theo1} can be extended to membrane grammars having more general rules, such as rules dividing non-elementary membranes and rules having polarizations or the possibility of changing the labels of membranes.


\end{document}